\documentclass[12pt]{article}
\usepackage{amsmath,amssymb,amsthm,graphicx,bm,geometry,setspace,color,enumitem,multirow}
\usepackage[authoryear,longnamesfirst]{natbib}
\newtheorem{theorem}{Theorem}
\newtheorem{lemma}{Lemma}

\newtheorem{assumption}{Assumption}
\newtheorem*{algorithm}{Algorithm}
 
\title{\Large Linear Programming Approach to Nonparametric Inference under Shape Restrictions: with an Application to Regression Kink Designs\thanks{
We benefited from very useful comments by Chris Taber.
We would like to thank Patty Anderson and Bruce Meyer for kindly agreeing to our use of the CWBH data.
H. D. Chiang is supported by the Office of the Vice Chancellor for Research and Graduate Education at UW–Madison with funding from the Wisconsin Alumni Research Foundation.
K. Kato is partially supported by NSF grants DMS-1952306 and DMS-2014636.
The usual disclaimer applies.}}
\author{Harold D. Chiang\thanks{H.D. Chiang: Department of Economics, University of Wisconsin - Madison, William H. Sewell Social Science Building, 1180 Observatory Drive Madison, WI 53706-1393. Email: hdchiang@wisc.edu}
\and
Kengo Kato\thanks{K. Kato: Department of Statistics and Data Science, Cornell University, 1194 Comstock Hall, Ithaca, NY 14853. Email: kk976@cornell.edu}
\and 
Yuya Sasaki\thanks{Y. Sasaki: Department of Economics, Vanderbilt University, VU Station B \#351819, 2301 Vanderbilt Place, Nashville, TN 37235-1819. Email: yuya.sasaki@vanderbilt.edu}
\and 
Takuya Ura\thanks{T. Ura: Department of Economics, University of California, Davis, 1151 Social Sciences and Humanities, Davis, CA 95616. Email: takura@ucdavis.edu}}
\date{}
\begin{document}
\maketitle
\begin{abstract}
We develop a novel method of constructing confidence bands for nonparametric regression functions under shape constraints. This method can be implemented via a linear programming, and it is thus computationally appealing. We illustrate a usage of our proposed method with an application to the regression kink design (RKD). Econometric analyses based on the RKD often suffer from wide confidence intervals due to slow convergence rates of nonparametric derivative estimators. We demonstrate that economic models and structures motivate shape restrictions, which in turn contribute to shrinking the confidence interval for an analysis of the causal effects of unemployment insurance benefits on unemployment durations.
\begin{description}
\item[Keywords:] linear programming, regression kink design, shape restriction, nonparametric inference, confidence band.
\item[JEL Classification:] C13, C14, C21
\end{description}
\end{abstract}

\newpage
\section{Introduction}
Nonparametric inference under shape restrictions is often computationally demanding.
For instance, inference based on test inversion would require a grid search over a high-dimensional sieve parameter space.
In this paper, we propose a computationally attractive method for nonparametric inference about regression functions under shape restrictions.
Notably, our method can be implemented via a linear programming, despite the complicated nature of nonparametric inference under shape restrictions. 

In many applications, economic structures often motivate shape restrictions, and such restrictions may contribute to delivering more informative statistical inference about the economic structure and causal effects.
We highlight a case in point in the context of the regression kink design \citep[RKD;][]{nielsen2010estimating,card2015inference,dong2016jump}.
Estimation and inference in the RKD rely on derivative estimators of nonparametric regression functions, which typically suffer from slow convergence rates and thus may lead to wide confidence intervals.
On the other hand, there are often natural and economically motivated restrictions in the levels and slopes of the regression function to the left and/or right of the kink location, and they can contribute to shrinking the lengths of the confidence interval.
In the context of the regression discontinuity design, \cite{armstrong2015adaptive} and \cite{babii2019isotonic} suggest usage of shape restrictions with related motivations.
The benefits of shape restrictions may well be even greater for the RKD than for the regression discontinuity design due to the slower convergence rates of the RKD estimators.

We are far from the first to study the problem of nonparametric inference under shape restrictions.
\cite{dumbgen2003optimal}, \cite{cai2013adaptive}, \cite{armstrong2015adaptive}, \cite{chernozhukov2015constrained}, \cite{horowitz2017nonparametric}, \cite{chen2018shape}, \cite{freyberger2018inference}, \cite{mogstad2018using}, \cite{fang2019general}, and \cite{zhu2020inference}, among others, propose various approaches to nonparametric inference under shape restrictions.
See \citet{chetverikov2018econometrics} and the journal issue edited by \citet{samworth2018special} for a comprehensive review of the related literature.
We advance the frontier of this literature by providing a computationally attractive approach.
Specifically, we provide a novel method of constructing confidence bands/regions/intervals whose boundaries can be fully characterized as solutions to linear programs.

This paper is closely related to \cite{freyberger2015identification}, who have considered a linear programming approach to inference under shape restrictions. 
Specifically, they propose a linear programming approach to inference about linear functionals of finite-dimensional parameters, where the parameter values are the values of the regression function evaluated at finite support points.\footnote{\cite{fang2020inference} also propose a linear programming approach to inference for a growing number of linear systems, although their focus is different from nonparametric regression functions under shape restrictions as in this paper.}
On the other hand, as acknowledged in \citet{freyberger2015identification}, ``[t]he use of shape restrictions with continuously distributed variables is beyond the scope of'' their paper.
We contribute to this literature by accommodating (discretely or continuously) infinite-dimensional parameters.
This extended framework allows for analysis of nonparametric regressions with infinitely supported (discrete or continuous) regressors, which are relevant to many applications including the regression discontinuity and kink designs among others. 

Our proposed inference procedure works as follows.
First, we use the sieve approximation \citep[cf.][]{chen2007large} of the nonparametric regression function.
We then construct a supremum test statistic as a linear function of the sieve parameters, compute its critical value by applying \citet{chernozhukov2017}, and then translate their relation into an inequality constraint.
Subject to this inequality constraint, together with the additional linear-in-sieve-parameter inequality constraints stemming from shape restrictions, we find the lower (respectively, upper) bound of the confidence band/interval by the minimizing (respectively, maximizing) the sieve representation with respect to the sieve parameters.
In the final step, we inflate the bounds by a sieve approximation error bound similarly to \citet{armstrong2018finite,armstrong2020simple}, \citet{noack2019bias}, \citet{schennach2020bias}, and \citet{katosasakiura2020}.

The rest of this paper is organized as follows.
Section \ref{sec:model} presents the model and an overview of the proposed procedure.
Section \ref{sec:theory} presents the size control.
Section \ref{sec:projection} describes the procedure when we are interested in a finite-dimensional linear feature of the regression function. 
Section \ref{sec:rkd} presents an application of the RKD, with detailed implementation procedures tailored to this application.
In an empirical application, we demonstrate that shape restrictions can shrink the lengths of the confidence interval. 
Section \ref{sec:conclusion} concludes.
Mathematical proofs and simulation analysis are collected in the appendix.

Throughout this paper, we assume that a data set $\{(Y_i,X_i^T):i=1,\dots,n\}$ consists of i.i.d. random vectors following the law of $(Y,X^T)$, where $Y$ is a real-value random variable and $X$ is a finite-dimensional random variable with the support $\mathcal{X}\subset\mathbb{R}^{\dim X}$. 
Let $E_n$ denote the sample mean, that is, $E_n[f(Y,X^T)]\equiv\frac{1}{n}\sum_{i=1}^nf(Y_i,X_i^T)$ for any measurable function $f$.

\section{Inference Method}\label{sec:model}
In this paper, we are interested in a linear feature of the unknown mean regression function $g_0(x)\equiv E[Y\mid X=x]$, so that the parameter of interest can be written as 
$$
\theta_0\equiv\mathbf{A}_0 g_0
$$ 
for a known linear operator $\mathbf{A}_0$.
We assume this parameter $\theta_0$ to be a function from some set $\mathcal{W}_0$ into $\mathbb{R}$, which allows $\theta_0$ to be a scalar, a vector, or a function from $\mathcal{X}$ into $\mathbb{R}$. 
For example, when $\mathbf{A}_0$ is the identity function, the parameter of interest is the conditional mean function $g_0$ itself. 
Other examples for $\theta_0$ include $g_0(x)$ for a given point $x$, the integral $\int g_0(x)d\mu(x)$, and the derivative $\partial g_0(x)/\partial x_j$, among others. 
In Section \ref{sec:projection}, we discuss how we can tailor the procedure to the case when $\theta_0$ is finite dimensional. 

The objective of this paper is to construct a confidence region for $\theta_0$ under the shape restrictions 
\begin{equation}\label{eq:shape_restriction}
[\mathbf{A}_1g_0] (w_1)\leq 0\mbox{ for every }w_1\in\mathcal{W}_1
\end{equation} 
for a known linear operator $\mathbf{A}_1$.\footnote{In this paper, the shape restriction does not have any improvement in the identification analysis, because  $g_0$ is identified over $\mathcal{X}$ and therefore $\theta_0$ is identified.} 
We are going to construct a confidence region ${CR}_{\theta}$ for $\theta_0$ satisfying the following two properties: 
(i) the boundaries of ${CR}_{\theta}$ are the set of solutions to linear programming problems; and 
(ii) ${CR}_{\theta}$ controls the asymptotic size under the shape restriction.

We approximate $g_0$ by a linear combination of $k$ functions $p_1,\ldots,p_k$ on $\mathcal{X}$.\footnote{Recall that $\mathcal{X}$ is the support of $X$. We assume $k\geq 2$, which guarantees $\log k\geq 0$.}
These $k$ functions are denoted by
$$
p_{1:k}\equiv (p_1,\ldots,p_k)^T.
$$ 
We can consider the linear regression of $Y$ on $p_{1:k}(X)$, and the population coefficient vector for this regression is 
$$
\bar{\beta}\equiv E[p_{1:k}(X)p_{1:k}(X)^T]^{-1}E\left[p_{1:k}(X)Y\right].
$$
With these definitions and notations, we make the following assumption about error bounds for the approximation of $g_0$ by $p_{1:k}^T\bar{\beta}$.

\begin{assumption}[Approximation error bounds]\label{assn:approx_error}
There exist known functions $\delta_0$ and $\delta_1$ such that
\begin{align}
&\left\vert [\mathbf{A}_0(g_0-p_{1:k}^T\bar{\beta})](w_0) \right\vert
\leq
\delta_0(w_0)
&&\text{ for all } w_0 \in \mathcal{W}_0;\mbox{ and }
\label{eq:app_bound_A0}
\\
&\left\vert [\mathbf{A}_1(g_0-p_{1:k}^T\bar{\beta})](w_1) \right\vert
\leq
\delta_1(w_1)
&&\text{ for all } w_1 \in \mathcal{W}_1.\label{eq:app_bound_A1}
\end{align}
\end{assumption}

This assumptions plays the role of restricting the function class where $g_0$ resides, similarly to \citet{katosasakiura2020} in the spirit of the honest inference approach \citep{armstrong2018finite,armstrong2020simple} and the bias bound approach \citep{schennach2020bias}.\footnote{We allow $k$, $\delta_0$ and $\delta_1$ to be a function of $n$. We do not require  $k\rightarrow\infty $ as $n\rightarrow \infty$ but it is allowed. 
	In Assumption \ref{assn:approx_error}, we bound the biases coming from the approximation of $g_0$ by $p_{1:k}^T\bar{\beta}$ by known $\delta_0$ and $\delta_1$. 
Without accounting for such approximation bounds, conventional methods would set $\delta_0\rightarrow 0$ and $\delta_1\rightarrow 0$ as $n\rightarrow 0$ in light of that the bias asymptotically vanishes with undersmoothing.
That said, by Assumption \ref{assn:approx_error}, we take this honest or bias bound approach in this paper for the sake of generality, with the special case of undersmoothing leading to the conventional approach in particular.}

For a generic value $\beta\in\mathbb{R}^k$, we can implement a hypothesis testing for the null hypothesis $H_0: \bar\beta=\beta$ against the alternative hypothesis $H_1: \bar\beta\ne\beta$ as follows.
In this hypothesis testing problem, we aim to detect a violation of the null hypothesis 
$$
H_0: E[p_{1:k}(X)(Y-p_{1:k}(X)^T\beta)]=0,
$$
which is equivalent to $\bar\beta=\beta$ under the invertibility of $E[p_{1:k}(X)p_{1:k}(X)^T]$.
We can estimate the left hand side of the above equation by $E_n[p_{1:k}(X)(Y-p_{1:k}(X)^T{\beta})]$ and its asymptotic variance (under $H_0$) by $E_n[\hat{\omega}\hat{\omega}^T]$, where 
$$
\hat{\omega}\equiv p_{1:k}(X)(Y-p_{1:k}(X)^T E_n\left[p_{1:k}(X)p_{1:k}(X)^T\right]^{-1}E_n\left[p_{1:k}(X)Y\right]).
$$
Note that $\hat{\omega}$ estimates $\omega\equiv p_{1:k}(X)(Y-p_{1:k}(X)^T\bar{\beta})$.
With these estimates, we consider the test statistic 
$$
\left\|E_n[\hat{\omega}\hat{\omega}^T]^{-1/2}E_n[p_{1:k}(X)(Y-p_{1:k}(X)^T{\beta})]\right\|_{\infty}.
$$ 
To obtain a critical value, we apply the multiplier bootstrap by calculating the $(1-\alpha)$ quantile, denoted by ${cv}$, of 
$$
\left\|E_n[\hat{\omega}\hat{\omega}^T]^{-1/2}E_n[\eta\hat{\omega}]
\right\|_{\infty}
$$ 
conditional on the data set, where $\eta_1,\ldots,\eta_n$ are independent Rademacher multiplier random variables that are independent of the data.
Note that the critical value ${cv}$ does not depend on a specific value of $\beta$, which enables us to construct a confidence region characterized by linear inequalities for $\beta$.

We can construct a confidence region for $\theta_0$ based on the test inversion. 
Using the test statistic and the critical value, we can define a confidence region  for $\theta_0$, denoted by ${CR}_{\theta}$. Namely, $CR_{\theta}$ is the set of $\theta$ satisfying the following linear constraints for some $\beta\in\mathbb{R}^k$: 
\begin{equation}\label{eq:sample_error_control}
\left\|E_n[\hat{\omega}\hat{\omega}^T]^{-1/2}E_n[p_{1:k}(X)(Y-p_{1:k}(X)^T{\beta})]\right\|_{\infty}\leq {cv},
\end{equation}
\begin{equation}\label{eq:approx_control}
|[\mathbf{A}_0p_{1:k}^T](w_0)\beta-\theta(w_0)|\leq\delta_0(w_0)\mbox{ for every }w_0\in\mathcal{W}_0,\mbox{ and }
\end{equation}
\begin{equation}\label{eq:additiona_restriction}
[\mathbf{A}_1p_{1:k}^T](w_1)\beta\leq\delta_1(w_1)\mbox{ for every }w_1\in\mathcal{W}_1,
\end{equation}
where $[\mathbf{A}_0p_{1:k}^T](w_1)\beta\equiv [\mathbf{A}_0(p_{1:k}^T\beta)](w_1)$ and $[\mathbf{A}_1p_{1:k}^T](w_1)\beta\equiv [\mathbf{A}_1(p_{1:k}^T\beta)](w_1)$.

In the definition of ${CR}_{\theta}$, we have three types of linear constraints. 
First, \eqref{eq:sample_error_control} comes from the hypothesis test for $H_0: \bar\beta=\beta$. 
Second, \eqref{eq:approx_control} controls the approximation error between $\mathbf{A}_0p_{1:k}^T\bar\beta$ and $\theta_0$ under \eqref{eq:app_bound_A0} in Assumption \ref{assn:approx_error}.
Third, \eqref{eq:additiona_restriction} uses the knowledge that the shape restriction \eqref{eq:shape_restriction} holds for true $g_0$, together with \eqref{eq:app_bound_A1} in Assumption \ref{assn:approx_error}.
This confidence region could be more informative than that without the shape-restriction inequalities in \eqref{eq:additiona_restriction}.

For every value $w_0 \in \mathcal{W}_0$, the following theorem states that the projection of ${CR}_{\theta}$ to $\theta_0(w_0)$ can be computed by solving two linear programming problems. A proof is provided in Appendix \ref{theorem:projection}.

\begin{theorem}\label{theorem:projection}
Under Assumption \ref{assn:approx_error}, for every $w_0 \in \mathcal{W}_0$, the projection of ${CR}_{\theta}$ to $\theta_0(w_0)$ is equal to the closed interval 
$$
\left[ \
\underset{s.t.\  \eqref{eq:sample_error_control}\&\eqref{eq:additiona_restriction}}{\min_{\beta}} \
[\mathbf{A}_0p_{1:k}^T](w_0) \beta-\delta_0(w_0), \ \ 
\underset{s.t.\  \eqref{eq:sample_error_control}\&\eqref{eq:additiona_restriction}}{\max_{\beta}} \
[\mathbf{A}_0p_{1:k}^T](w_0) \beta+\delta_0(w_0) \
\right].
$$
\end{theorem}
\noindent
Therefore, the boundary points are the solutions to linear programs.

\section{Size Control}\label{sec:theory}
For the asymptotic size control, we are going to impose the following assumptions. 
Let $b>0$, $q \in [4,\infty), \nu\in (2,\infty)$ be some constants and let $B_n\ge 1$ denote a sequence of finite constants that may possibly diverge to infinity. 
Consider the following assumption.

\begin{assumption} 
\label{assn:new_combined_all}
(a)
The eigenvalues of ${E[\omega\omega^T]}$ and $E[p_{1:k}(X)p_{1:k}(X)^T]$ are bounded above and bounded below away from $0$ uniformly over $n$.  
(b)
(i) $E[Y^2]<\infty$.  
(ii)  $E[|(E[\omega\omega^T]^{-1/2})_j\omega|^2]\ge b$, $E[|(E[\omega\omega^T]^{-1/2})_j\omega|^{2+\kappa}]\le B_n^\kappa$ and $E[ \|E[\omega\omega^T]^{-1/2}\omega\|_\infty^q]\le B_n^{q}$ for every $j=1,\ldots,k$ and each $\kappa=1,2$.\footnote{$(E[\omega\omega^T]^{-1/2})_j$ denotes the $j$-th row of a square matrix $E[\omega\omega^T]^{-1/2}$.} 
(iii) ${B_n^2 \log^7(nk)}/{n}=o(1)$ and ${B_n^2\log^3(nk)}/{n^{1-2/q}}=o(1)$.
(c)
(i) $\sup_{x\in\mathcal X}E[|Y-g_0(X)|^\nu| X=x]=O( 1)$.
(ii) For every $k$, there are finite constants $c_k$ and $\ell_k$ such that $E[(g_0(X)-p_{1:k}(X)^T\bar{\beta})^2]^{1/2}\le c_k$ and that $\sup_{x\in\mathcal X}|g_0(x)-p_{1:k}(x)^T\bar{\beta}|\le \ell_kc_k$.
(iii) Let $\xi_k\equiv\sup_{x\in\mathcal{X}}\|p_{1:k}(x)\|_2$ and $\xi_k^L\equiv \sup_{x,x'\in\mathcal X:\,x\ne x'}{\left\|p_{1:k}(x)/\|p_{1:k}(x)\|_2-p_{1:k}(x')/\|p_{1:k}(x')\|_2\right\|_2}/{\|x-x'\|_2}$. Then  $\xi_k^{2\nu/(\nu-2)}\log k/n= O(1)$, $\log \xi_k^L = O( \log k)$, and $\log \xi_k = O(\log k)$.
(iv) $n^{-1}\xi_k^2 \log k=o(1)$, $\ell_kc_k =O(1)$, and $(n^{-1}\xi_k^2)^{1/2}\left\{n^{1/\nu}(\log k)^{1/2} +\sqrt{k}\right\}=O( 1 )$. 
\end{assumption}

Assumption \ref{assn:new_combined_all} (a) implies Condition A.2 in Assumption \cite{belloni2015some}. 
It imposes a restriction to rule out overly strong co-linearity among $p_1,\dots,p_k$. 
Assumptions \ref{assn:new_combined_all} (b)-(ii) and \ref{assn:new_combined_all} (b)-(iii) correspond to Conditions (M.1), (M.2) and (E.2) in \cite{chernozhukov2017}. 
It requires that the polynomial moments of the maximal component of normalized $\omega$ will not be growing too fast, as well as it imposes conditions that dictate how fast the number of basis functions can grow. 
The maximum is allowed to be growing at a rate of $O(n^{a})$ for some $a$ between zero and one. 
Assumption \ref{assn:new_combined_all} (c) covers Conditions A.3-A.5 in \cite{belloni2015some} as well as rate conditions in the statement of their Theorem 4.6. 
Assumption \ref{assn:new_combined_all} (c)-(i) requires the residual to have a finite $\nu$-th moment for some $\nu>2$. 
Assumptions \ref{assn:new_combined_all} (c)-(ii) and \ref{assn:new_combined_all} (c)-(iii) impose bounds on the approximation errors of $g_0$ using $p_1,\dots,p_k$, as well as restrictions on the size of basis functions, measured by the Euclidean norm and the Lipschitz constant. Assumption \ref{assn:new_combined_all} (c)-(iv) imposes some more constraints on the relative growth rates of the approximation errors, the size and number of basis functions. 
Notice that it does not require the approximation errors to be diminishing asymptotically, and hence does not require undersmoothing. 
 
The following theorem states the asymptotic size control for ${CR}_{\theta}
$ as a confidence region for $\theta_0$. 
A proof is provided in Appendix \ref{theorem:main_size}.

\begin{theorem}\label{theorem:main_size}
If Assumptions \ref{assn:approx_error} and \ref{assn:new_combined_all} are satisfied,
then 
$$
\liminf_{n\rightarrow\infty}\mathbb{P}\left(\theta_0\in {CR}_{\theta}
\right)\geq 1-\alpha.
$$
\end{theorem}

With some additional notations and rate conditions, it is possible to strengthen the statement of Theorem \ref{theorem:main_size} to hold uniformly over a set of data generating processes. This is due to the fact that key theoretical building blocks in the proof of Theorem \ref{theorem:main_size} -- i.e. the anti-concentration inequality in \cite{chernozhukov2015comparison}, the high-dimensional central limit theorem of \cite{chernozhukov/chetverikov/kato:2017}, and Rudelson's concentration inequality \citep[Lemma 6.2]{belloni2015some} -- all provide non-asymptotic bounds with constants only depending on a few key features of the model such as $b,q$ and $\nu$.

\section{Inference Method for Finite Dimensional $\theta_0$}\label{sec:projection}

When the parameter of interest $\theta_0$ is finite dimensional, we can directly test $\mathbf{A}_0[p_{1:k}^T\bar\beta]=\theta$ for a generic value of $\theta$, instead of testing $\bar\beta=\beta$ as in Section \ref{sec:model}. 
In the current section, we describe the inference procedure when $\theta_0$ is a finite-dimensional column vector.

For a generic value $\theta$, we consider the null hypothesis $H_0: A_{0,k}\bar\beta=\theta$ and the alternative hypothesis $H_1: A_{0,k}\bar\beta\ne\theta$, where $A_{0,k}$ is the matrix defined by $A_{0,k}\beta=\mathbf{A}_0[p_{1:k}^T\beta]$ for every $k\times 1$ vector $\beta$. 
Based on the definition of $\bar\beta$, we aim to measure the violation of the null hypothesis 
$$
H_0: A_{0,k}E[p_{1:k}(X)p_{1:k}(X)^T]^{-1}E\left[p_{1:k}(X)Y\right]=\theta.
$$
We can estimate the left hand side by $A_{0,k}E_n\left[p_{1:k}(X)p_{1:k}(X)^T\right]^{-1}E_n\left[p_{1:k}(X)Y\right]$ and its the asymptotic variance under $H_0$ by 
$$
\hat{V}\equiv A_{0,k}E_n\left[p_{1:k}(X)p_{1:k}(X)^T\right]^{-1}E_n[\hat{\omega}\hat{\omega}^T]E_n\left[p_{1:k}(X)p_{1:k}(X)^T\right]^{-1}A_{0,k}^T.
$$
With these estimators, we consider the test statistic 
$$
\left\|
\hat{V}^{-1/2}(A_{0,k}E_n\left[p_{1:k}(X)p_{1:k}(X)^T\right]^{-1}E_n\left[p_{1:k}(X)Y\right]-\theta)
\right\|_{\infty}.
$$
To obtain its critical value, we apply the multiplier bootstrap and compute the $(1-\alpha)$ quantile, denoted by $\widehat{cv}$, of 
$$
\left\|
\hat{V}^{-1/2}A_{0,k}E_n\left[p_{1:k}(X)p_{1:k}(X)^T\right]^{-1}E_n[\eta\hat{\omega}]
\right\|_{\infty}
$$
conditional on the data set, where $\eta_1,\ldots,\eta_n$ are independent Rademacher multiplier random variables that are independent of the data.

A confidence region for $\theta_0$ can be constructed based on the test inversion. 
In this setup, we can construct a confidence region for $\theta_0$, $\widehat{CR}_{\theta}$, by collecting all $\theta$'s satisfying the following linear constraints for some $\beta\in\mathbb{R}^k$: 
\begin{equation}\label{eq:sample_error_control_proj}
\left\|
\hat{V}^{-1/2}(A_{0,k}E_n\left[p_{1:k}(X)p_{1:k}(X)^T\right]^{-1}E_n\left[p_{1:k}(X)Y\right]-A_{0,k}\beta)
\right\|_{\infty}\leq \widehat{cv},
\end{equation}
$$
|[\mathbf{A}_0p_{1:k}^T](w_0)-\theta(w_0)|\leq\delta_0(w_0)\mbox{ for every }w_0\in\mathcal{W}_0,\mbox{ and }
$$
\begin{equation}\label{eq:additiona_restriction_proj}
[\mathbf{A}_1p_{1:k}^T](w_1)\beta\leq\delta_1(w_1)\mbox{ for every }w_1\in\mathcal{W}_1.
\end{equation}

For every value $w_0 \in \mathcal{W}_0$, we can compute the projection of $\widehat{CR}_{\theta}$ to $\theta_0(w_0)$ by solving two linear programming problems w.r.t. $\beta$:
$$
\mbox{ minimize }[A_{0,k}\beta](w_0)-\delta_0(w_0)\mbox{ over $\beta$ subject to } \eqref{eq:sample_error_control_proj} \ \& \ \eqref{eq:additiona_restriction_proj},
$$
and 
$$ 
\mbox{ maximize }[A_{0,k}\beta](w_0)+\delta_0(w_0)\mbox{ over $\beta$ subject to } \eqref{eq:sample_error_control_proj} \ \& \ \eqref{eq:additiona_restriction_proj}.
$$
In other words, the projection is the closed interval 
$$
\left[\
\underset{\text{s.t.}\ \eqref{eq:sample_error_control_proj}\&\eqref{eq:additiona_restriction_proj}}{\min_{\beta}}
[A_{0,k}\beta](w_0)-\delta_0(w_0),\ \
\underset{\text{s.t.}\ \eqref{eq:sample_error_control_proj}\&\eqref{eq:additiona_restriction_proj}}{\max_{\beta}}
[A_{0,k}\beta](w_0)+\delta_0(w_0)\
\right].
$$

Formal theoretical properties of the the confidence interval constructed by this procedure follow from analogous arguments to those in Sections \ref{sec:model} and \ref{sec:theory}.
In the application presented in the following section, the parameter $\theta_0$ of interest is a scalar (and finite dimensional in particular) and we therefore adopt this approach to constructing its confidence interval.

\section{Application to Regression Kink Design}\label{sec:rkd}
In this section, we present an application of our proposed method to the regression kink design (RKD). 
Since the regression kink design is based on estimates of slopes as opposed to levels, statistical inference based on nonparametric estimates often entails slow convergence rates and thus wide confidence intervals.
To mitigate this adverse feature of the regression kink design, we propose to impose shape restrictions that are motivated by the underlying economic structures.

To introduce the RKD, consider the structure
$$
Y=Y(T,X,U)\mbox{ and }T= T(X),
$$
where $Y$ denotes the outcome variable, $T$ denotes the treatment variable, $X$ denotes the running variable, and $U$ denotes the random vector of unobserved characteristics.
A researcher is often interested in the partial effect $\partial Y(T,X,U) / \partial T$ of the treatment variable on the outcome variable.
Since the unobserved characteristics $U$ are generally correlated with the running variable $X$ and thus with the treatment $T=T(X)$, one would need to exploit exogenous variations in the treatment variable in order to identify this partial effect.
If the treatment policy function $T(\cdot)$ exhibits a `kink' at a known point $\bar{x}$, then this shape restriction can be exploited to induce local exogenous variations in the treatment variable $T$ as well, so that the partial effect of interest may be identified.
This approach of the so-called regression kink design (RKD) was proposed by \citet{nielsen2010estimating} and \citet{card2015inference} -- see \citet{dong2016jump} for the case of a binary treatment, and see \citet{chiang2019causal} and \citet{chen2020quantile} for heterogeneous treatment effects.

Suppose that a researcher is interested in conducting inference for the average partial effect $h^1(\bar{x})\equiv E\left[\left. {\partial Y(T,X,U)}/{\partial T} \right\vert X=\bar{x}\right]$ at  the kink point $\bar{x}$.
Under regularity conditions, we can obtain the following decomposition of the derivative $g_0'(X)$ of $g_0(x) = E[Y|X=x]$:
\begin{align}
g_0'(x)
&=
\underbrace{E\left[\left. \frac{\partial Y(T,X,U)}{\partial T} \right\vert X=x\right]}_{\text{Partial Effect of Interest: } h^1(x)} \cdot T'(x)
+
\underbrace{E\left[\left. \frac{\partial Y(T,X,U)}{\partial X} \right\vert X=x\right]}_{\text{Direct Effect of $X$: } h^2(x)}
\notag\\
&\quad+ 
\underbrace{E\left[\left. Y \cdot \frac{\partial \log f_{U|X}(U|X)}{\partial X} \right\vert X=x\right]}_{\text{Endogenous Effect: } h^3(x)}.
\label{eq:rkd_decomposition}
\end{align}
If $T'(\cdot)$ is discontinuous (i.e., $T(\cdot)$ is kinked) at $\bar{x}$ while each of $h^1$, $h^2$ and $h^3$ is continuous at $\bar{x}$, then this decomposition implies that the partial effect of interest at $\bar{x}$ can be identified by
$$
h^1(\bar{x}) = \frac{\lim_{x\downarrow\bar{x}} g_0'(x) - \lim_{x\uparrow\bar{x}} g_0'(x)}{\lim_{x\downarrow\bar{x}} T'(x) - \lim_{x\uparrow\bar{x}} T'(x)},
$$
cf. \citet{nielsen2010estimating,card2015inference}.
We can represent the parameter of interest via $h^1(\bar{x})=\mathbf{A}_0 g_0$, using a linear operator $\mathbf{A}_0$ defined by 
\begin{equation}\label{eq:RKD_A_0}
\mathbf{A}_0 g =\frac{\lim_{x\downarrow\bar{x}} g'(x) - \lim_{x\uparrow\bar{x}} g'(x)}{\lim_{x\downarrow\bar{x}} T'(x) - \lim_{x\uparrow\bar{x}} T'(x)}.
\end{equation}
Even though $g_0$ is unknown, the operator  $\mathbf{A}_0$ is known since $T(\cdot)$ is a known function.
In this case, $\mathcal{W}_0=\{\bar{x}\}$, and the parameter of interest $\theta_0=\mathbf{A}_0 g_0$ is a scalar.

Although $\theta_0$ is nonparametrically estimable, an estimator based on slopes of nonparametric regression functions usually suffers from slow rates of convergence, and thus it may not provide an informative confidence interval.
If an economic structure motivates shape restrictions, then imposing such restrictions may conceivably contribute to shrinking the length of the confidence interval.
With this motivation, in Section \ref{sec:empirical_analysis}, we demonstrate how shape restrictions help in conducting statistical inference in the analysis of of unemployment insurance (UI).

\subsection{Causal Effects of UI Benefits on Unemployment Duration}\label{sec:empirical_analysis}

Unemployment insurance (UI) benefits play important roles in supporting consumption smoothing under the risk of unemployment.
A potential drawback of the UI benefits is the moral hazard effects, that is, the UI benefits may discourage unemployed workers from looking for jobs, leading to elongated unemployment durations and thus economic inefficiency.
Identifying and estimating these moral hazard effects have been of research interest in labor economics. 
\citet{landais2015assessing} suggests to exploit the non-smooth UI benefit schedule as detailed below, and thus to use the regression kink design to identify the effects of UI benefits on the duration of unemployment.
Applying this identification strategy to the data of the Continuous Wage and Benefit History Project \citep[cf.][]{moffitt1985effect}, \citet{landais2015assessing} finds that there are positive effects of the UI benefit amounts on the duration of unemployment, even after controlling for unobserved source of endogenous selection of the duration that may be correlated with the pre-unemployment income and thus the benefit amount.
\citet{chiang2019causal} further investigate heterogeneous effects of the UI benefit amount on the duration by using the quantile regression kink design.

\cite{landais2015assessing} considers the following empirical framework of assessing the welfare effects of unemployment benefits.
The outcome $Y$ of interest is the duration of unemployment.
Upon becoming unemployed, an individual can apply for UI and receives a weekly benefit amount of $T=T(X)$, where $X$ is the highest quarterly earning in the last four completed calendar quarters prior to the date of the UI claim.
The partial effect $\partial Y(T,X,U) / \partial T$ measures the moral hazard effect of the UI benefits on the duration of unemployment in this setting.
Since the unobserved characteristics $U$ contain cognitive and non-cognitive skills of the individual, such as attitudes toward work, that are generally correlated with the labor income $X$ received prior to the unemployment, one would need exogenous variations in the treatment variable in order to identify this moral hazard effect.

As in \citet{landais2015assessing}, we can exploit the fact that the UI benefits policy $T(\cdot)$ exhibits a kinked shape.
In particular, the UI schedule in the state of Louisiana is linear in $X$ with a constant $t\equiv 1/25$ of proportionality up to a fixed ceiling $t_{\max}$.
(Note that the unit of $X$ is U.S. dollars per quarter, whereas the unit of $T(X)$ is U.S. dollars per week. Therefore, this constant of proportionality implies that the UI benefit amount is approximately a half of the prior earnings.)
The maximum UI benefit amount is $\bar{t}=$ \$183 during the period between September 1981 and September 1982, and $\bar{t}=$ \$205 during the period between September 1982 and December 1983.
In short, the UI benefits policy  takes the form of 
$$
T(x) =
\begin{cases}
t \cdot x & \text{if } x < t_{\max} / t\\
t_{\max}  & \text{if } x \ge t_{\max} / t,
\end{cases}
$$
and $T$ is thus kinked at $\bar{x}= t_{\max}/t$. 
Individuals can continue to receive the benefits determined by this formula as far as they remain unemployed up to the maximum duration of 28 weeks.

We construct a data set by following the data construction in \cite{landais2015assessing} and \cite{chiang2019causal}.
We focus on the observations in Louisiana. 
The sample size of the original data is 9,008 for the period between September 1981 and September 1982, and 16,463 for the period between September 1982 and December 1983.
Since we are interested in the information around the kink location $\bar{x}$, for simplicity, we focus on the (sub-)sample of the observations in the interval $X\in[\bar{x}-5000,\bar{x}+5000]$.
The resultant sample size is 8,677 for the period between September 1981 and September 1982, and the resultant sample size is 15,763 for the period between September 1982 and December 1983.

In this empirical application, we can consider a few shape restrictions on the unknown conditional mean function $g_0(x)=E[Y\mid X=x]$.
First of all, to impose the continuity of $g_0$ at $\bar{x}$, we can use the shape restriction 
\begin{equation}
\lim_{x\downarrow\bar{x}} g_0(x)=\lim_{x\uparrow\bar{x}} g_0(x).
\label{eq:rkd_shape_restriction_A}
\end{equation}
This restriction is not redundant when we use difference sieves for the left of $\bar{x}$ and the right of $\bar{x}$.
Moreover, it may be reasonable to assume that  $h^2$ and $h^3$ are both non-increasing.
Specifically, the direct effect $h^2$ is non-increasing if formerly higher-income earner can find the next job more quickly than formerly lower-income earners on average.
The endogenous effect $h^3$ is non-increasing if individuals with higher abilities can find the next job more quickly than those with lower abilities on average.
Since $T(\cdot)$ is a constant function to the right of the kink location in this application, this assumption together with the decomposition \eqref{eq:rkd_decomposition} implies that the reduced form $g_0$ is non-increasing to the right of the kink location $\bar{x}$.
This consideration leads to the slope restriction 
\begin{equation}
g_0'(x)\leq\mbox{ for every }x>\bar{x}.
\label{eq:rkd_shape_restriction_B}
\end{equation}
In the notations in Section \ref{sec:model}, we can summarize the shape restrictions \eqref{eq:rkd_shape_restriction_A} and  \eqref{eq:rkd_shape_restriction_B} as 
\begin{equation}
[\mathbf{A}_1 g_0] (w_1)\leq 0\mbox{ for every }w_1\in\mathcal{W}_1,
\label{eq:rkd_shape_restriction_2}
\end{equation}
where $\mathcal{W}_1=\{-2,-1\}\cup\{w_1:w_1>\bar{x}\}$ and 
$$
[\mathbf{A}_1 g] (w_1) =
\begin{cases}
\lim_{x\downarrow\bar{x}} g(x) - \lim_{x\uparrow\bar{x}} g(x)&\mbox{ if }w_1=-2\\
\lim_{x\uparrow\bar{x}} g(x) - \lim_{x\downarrow\bar{x}} g(x)&\mbox{ if }w_1=-1\\
g'(w_1)&\mbox{ if }w_1>\bar{x}.
\end{cases}
$$

Now, we outline the concrete implementation procedure to exploit these shape restrictions \eqref{eq:rkd_shape_restriction_2}, for inference about the causal parameter $\theta_0 = \mathbf{A}_0 g_0$ defined in \eqref{eq:RKD_A_0}.
For every even natural number $k$, we use the basis functions  
$$
p_{1:k}=(\ell_{L,0},\ell_{R,0},\cdots,\ell_{L,k/2-1},\ell_{R,k/2-1}),
$$ 
where $\left(\ell_{L,0},\ell_{L,1},\cdots,\ell_{L,k/2-1}\right)$ are the first $k/2$ terms of an orthonormal basis for $L^2([\bar{x}-5000,\bar{x}])$ and $\left(\ell_{R,0},\ell_{R,1},\cdots,\ell_{R,k/2-1}\right)$ are the first $k/2$ terms of an orthonormal basis for $L^2([\bar{x},\bar{x}+5000])$.
We use the shifted Legendre bases in the empirical application in this subsection as well as in the simulation studies in Section \ref{sec:simulation}.
We follow Section \ref{sec:projection} to construct the $(1-\alpha)$-level confidence interval for $\theta_0$ subject to the shape constraint \eqref{eq:rkd_shape_restriction_2}, where we restrict $\mathcal{W}_1=\{-2,-1\}\cup\{\xi_1,\dots,\xi_l\}$ with 99 equally spaced grid points $\{\xi_1,\dots,\xi_l\} \subset (\bar{x},\bar{x}+5000)$.
The following algorithm provides a step-by-step procedure of the construction.
\begin{algorithm}
\begin{enumerate}
\item[]
\item 
For every observation $i=1,\ldots,n$, construct the vector 
$$
\hat\omega_i = p_{1:k}(X_i)\left( Y_i - p_{1:k}(X_i)^T E_n\left[p_{1:k}(X)p_{1:k}(X)^T\right]^{-1} E_n\left[p_{1:k}(X)Y\right] \right).
$$
\item Construct the four matrices: 
{\small
\begingroup
\allowdisplaybreaks
\begin{align*}
A_{0,k}&=\left(\begin{array}{ccccc}
-\lim_{x\uparrow\bar{x}} \ell_{L,0}'(x) & \lim_{x\downarrow\bar{x}} \ell_{R,0}'(x) & \cdots & -\lim_{x\uparrow\bar{x}} \ell_{L,k/2-1}'(x) & \lim_{x\downarrow\bar{x}} \ell_{R,k/2-1}'(x)
\end{array}\right),
\\
B_0&=\frac{A_{0,k}E_n\left[p_{1:k}(X)p_{1:k}(X)^T\right]^{-1}E_n[p_{1:k}(X)Y]}{\sqrt{A_{0,k}E_n\left[p_{1:k}(X)p_{1:k}(X)^T\right]^{-1}E_n[\hat{\omega}\hat{\omega}^T]E_n\left[p_{1:k}(X)p_{1:k}(X)^T\right]^{-1}A_{0,k}^T}},
\\
B_1&=\frac{A_{0,k}}{\sqrt{A_{0,k}E_n\left[p_{1:k}(X)p_{1:k}(X)^T\right]^{-1}E_n[\hat{\omega}\hat{\omega}^T]E_n\left[p_{1:k}(X)p_{1:k}(X)^T\right]^{-1}A_{0,k}^T}},\mbox{ and }
\\
B_2&=
\left(\begin{array}{ccccc}
-\lim_{x\uparrow\bar{x}} \ell_{L,0}(x) & \lim_{x\downarrow\bar{x}} \ell_{R,0}(x) & \cdots & -\lim_{x\uparrow\bar{x}} \ell_{L,k/2-1}(x) & \lim_{x\downarrow\bar{x}} \ell_{R,k/2-1}(x)
\\
\lim_{x\uparrow\bar{x}} \ell_{L,0}(x) & -\lim_{x\downarrow\bar{x}} \ell_{R,0}(x) & \cdots & \lim_{x\uparrow\bar{x}} \ell_{L,k/2-1}(x) & -\lim_{x\downarrow\bar{x}} \ell_{R,k/2-1}(x)
\\
0 & \ell_{R,0}'(\xi_1) & \cdots & 0 & \ell_{R,k/2-1}'(\xi_1)
\\
\vdots&\vdots&&\vdots&\vdots
\\
\\
0 & \ell_{R,0}'(\xi_l) & \cdots & 0 & \ell_{R,k/2-1}'(\xi_l)
\end{array}\right).
\end{align*}
\endgroup
}
\item
Generate $M$ independent samples $\{\eta_{m,1},\cdots,\eta_{m,n}\}_{m=1,\ldots,M}$ of Rademacher random variables independently from data, and compute $\widehat{cv}$ by the $(1-\alpha)$-quantile of 
$$
\left\{\frac{\left|A_{0,k}E_n\left[p_{1:k}(X)p_{1:k}(X)^T\right]^{-1}E_n[\eta_m\hat{\omega}]\right|}{\sqrt{A_{0,k}E_n\left[p_{1:k}(X)p_{1:k}(X)^T\right]^{-1}E_n[\hat{\omega}\hat{\omega}^T]E_n\left[p_{1:k}(X)p_{1:k}(X)^T\right]^{-1}A_{0,k}^T}}\right\}_{m=1,\ldots,M}.
$$ 
\item Solve the linear programs
\begin{align*}
\min_{\beta} \ & \ A_{0,k}\beta- \delta_0 &&& \max_{\beta} \ & \ A_{0,k}\beta+ \delta_0
\\
\ \text{s.t.} \ & \ B_1 \beta \le B_0 + {cv} &&& \ \text{s.t.} \ & \ B_1 \beta \le B_0 +\hat{cv}
\\
& \ B_1 \beta \ge B_0 - {cv} &&&& \ B_1 \beta \ge B_0 - \hat{cv}
\\
& \ B_2 \beta \le \delta_1 &&&& \ B_2 \beta \le \delta_1.
\end{align*}
The solutions to these two linear programs are the boundary points of the $(1-\alpha)$-level confidence interval for $\theta_0$.
\end{enumerate}
\end{algorithm}

Table \ref{tab:empirical} summarizes the results for the statistical inference about the marginal effects of UI benefits on unemployment duration in Louisiana, based on the above algorithm.
Displayed are the 95\% confidence intervals and their lengths for each of the period between September 1981 and September 1982 (top panel) and the period between September 1982 and December 1983 (bottom panel).
We use the largest sieve dimension $k=12$ among those that were used in our simulation studies presented in Appendix \ref{sec:simulation}.
(The shape restrictions do not bind for the cases of $k=4$ or $k=8$. It is possibly because the current sample sizes are much larger than those used in our simulation studies.)
For the UI benefit amount $T(X)$, we use two alternative measures. 
One is the amount of UI benefits claimed (left half of each panel) and the other is the amount of UI benefits actually paid (right half of each panel) by following the prior work.
That said, these two alternative measures provide almost the same results, and therefore our discussions below apply to the results based on both of the two measures.

The reported confidence intervals contain the point estimates reported in the prior work by \citet{landais2015assessing}.
That said, the econometric specifications are different, and results are thus hard to compare.
Our results based on no shape restriction are effectively what we would get from the standard method with running the fifth-degree polynomial regressions on each side of the left and right of $\overline{x}$.
In contrast, \citet{landais2015assessing} uses the polynomials of degree one, i.e., the linear specification, for the main estimation results reported in his Table 2.
Due to the greater flexibility of our econometric specification, our method naturally incurs wider confidence intervals, but we demonstrate that shape restrictions will contribute to providing more informative results.

Our confidence interval includes the zero for the period between September 1981 and September 1982 (the first panel of Figure \ref{tab:empirical}) if no shape restriction is imposed, i.e., if the conventional approach is taken.
However, in this panel (for the period between September 1981 and September 198), shape restrictions \eqref{eq:rkd_shape_restriction_2} shrink the confidence intervals.
(Although these shrunken confidence intervals have their lower bounds approximately 0.000, note that we do not directly impose a sign restriction on the causal effects \textit{per se}, in the shape restrictions \eqref{eq:rkd_shape_restriction_2}. See our discussions above \eqref{eq:rkd_shape_restriction_2} for motivations of these shape restrictions.)
On the other hand, the confidence intervals are already informative for the period between September 1981 and September 1982 even without any shape restriction, and imposing shape restrictions \eqref{eq:rkd_shape_restriction_2} therefore will not contribute to shrinking the confidence intervals.
These results thus demonstrate one case in which shape restrictions contribute to enhancing the informativeness of statistical inference, and another case in which they do not.

\begin{table}[t]
	\centering
		\scalebox{0.95}{
		\begin{tabular}{rrlccrlc}
		\hline\hline
			September 1981 -- September 1982
			& \multicolumn{3}{c}{UI Claimed} && \multicolumn{3}{c}{UI Paid} \\
		\cline{2-4}\cline{6-8}
			Sieve Dimension: $k=12$
			& \multicolumn{2}{c}{95\% CI} & Length && \multicolumn{2}{c}{95\% CI} & Length\\
			\hline
			No Shape Restriction & 
		 [-0.023, & 0.044] & 0.067 &&[-0.030, & 0.040] & 0.070\\
			Shape Restrictions \eqref{eq:rkd_shape_restriction_2} & 
		  [0.000, & 0.044] & 0.044 && [0.000, & 0.040] & 0.040\\
		\hline\hline	
		\\\\
		\hline\hline
			September 1982 -- December 1983
			& \multicolumn{3}{c}{UI Claimed} && \multicolumn{3}{c}{UI Paid} \\
		\cline{2-4}\cline{6-8}
			Sieve Dimension: $k=12$
			& \multicolumn{2}{c}{95\% CI} & Length && \multicolumn{2}{c}{95\% CI} & Length\\
			\hline
			No Shape Restriction & 
			[0.002, & 0.048] & 0.046 && [0.002, & 0.047] & 0.045\\
			Shape Restrictions \eqref{eq:rkd_shape_restriction_2} & 
			[0.002, & 0.048] & 0.046 && [0.002, & 0.047] & 0.045\\
		\hline\hline	
		\end{tabular}
		}
	\caption{95\% confidence intervals of the marginal effect of UI benefit amount on unemployment duration for Louisiana, 1981--1983.}
	\label{tab:empirical}
\end{table}

\section{Conclusion}\label{sec:conclusion}

Nonparametric inference under shape restrictions can demand high computational burdens, e.g., a grid search over a high-dimensional sieve parameter space.
In this paper, we provide a novel method of constructing confidence bands/intervals for nonparametric regression functions under shape constraints.
The proposed method can be implemented via a linear programming, and it thus relieves the conventional computationally burdens.
A usage of this new method is illustrated with an application to the regression kink design.
Inference in the regression kink design often suffers from wide confidence intervals due to the slow convergence rates of nonparametric derivative estimators.
If economic models and structures motivate shape restrictions, then these restrictions may contribute to shrinking the confidence interval.
We demonstrate this point with real data for an analysis of the causal effects of unemployment insurance benefits on unemployment durations.
Specifically, for analysis of the effects of unemployment insurance benefits on the unemployment duration, the shape restrictions motivated by non-increasing direct effects and non-increasing endogenous effects drastically shrink the confidence interval of causal effects.

\newpage
\appendix
\section*{Appendix}

\section{Proofs for the Results in the Main Text}\label{sec:proofs}
\subsection{Proof of Theorem \ref{theorem:projection}}\label{sec:theorem:projection}

\begin{proof}
First, we are going to show that the projection of $CR_{\theta}$ to $\theta_0(w_0)$ is included in the interval defined in Theorem \ref{theorem:projection}.
Let $\theta$ be any element of $CR_{\theta}$. 
Then $[A_{0,k}\beta](w_0)-\delta_0(w_0)\leq\theta(w_0)\leq [A_{0,k}\beta](w_0)+\delta_0(w_0)$ for some $\beta\in\mathbb{R}^k$ such that \eqref{eq:sample_error_control} and \eqref{eq:additiona_restriction}.
It implies $\theta(w_0)$ is included in the interval.

Then, we are going to show that the interval is included in the projection of $CR_{\theta}$ to $\theta_0(w_0)$. 
Let $c$ be any element of the interval defined in Theorem \ref{theorem:projection}. 
There is $\beta$ such that $|[A_{0,k}\beta](w_0)-c|\leq\delta_0(w_0)$ and that $\beta$ satisfies \eqref{eq:sample_error_control} and \eqref{eq:additiona_restriction}. 
Define $\theta(\tilde{w}_0)$ by setting it to $[A_{0,k}\beta](\tilde{w}_0)$ for $\tilde{w}_0\ne w_0$ and to $c$ for $w_0$.  
Then this $\theta$ satisfies \eqref{eq:approx_control} with $\theta(w_0)=c$. 
It implies $c$ is included in the projection of $CR_{\theta}$ to $\theta_0(w_0)$. 
\end{proof}

\subsection{Proof of Theorem \ref{theorem:main_size}}\label{sec:theorem:main_size}

We first state four lemmas that play important roles in the proof of Theorem \ref{theorem:main_size}. Their proofs are delegated to Appendix \ref{sec:proofs_auxiliary_lemmas}. 
\begin{lemma}\label{lemma:GaussApprox}
Under Assumptions \ref{assn:new_combined_all} (a) and \ref{assn:new_combined_all} (b), 
there exist $k$-dimensional centered Gaussian random vectors $Z$ and $Z^\ast$ such that 
\begin{align*}
&\sup_{t}
\left|\mathbb{P}\left(\left\|Z\right\|_{\infty}\leq t\right)-\mathbb{P}\left(\left\|E_n[{E[\omega\omega^T]}^{-1/2}\omega]
\right\|_{\infty}\leq t\right)\right|=o(1),\\
&\sup_{t}\left|\mathbb{P}\left(\left\|Z^\ast\right\|_{\infty}\leq t\right)
-\mathbb{P}\left(\left\|E_n[{E[\omega\omega^T]}^{-1/2}\eta{\omega}]
\right\|_{\infty}\leq t\right)\right|=o(1),
\end{align*}
and
$
E[ZZ^T]=E[Z^\ast(Z^\ast)^T].
$
\end{lemma}
\begin{lemma}\label{lemma:conv_omega2}
Under Assumptions \ref{assn:new_combined_all} (a) and \ref{assn:new_combined_all} (b), we have 
$$
\max\left\{\left\|E_n[(\eta+1){\omega}]\right\|_2,\left\|E_n[{\omega}]\right\|_2\right\}=O_{P}\left(\sqrt{\frac{\xi_k^2}{n}}\right).
$$
\end{lemma}

\begin{lemma}\label{lemma:conv_Qast}
Under Assumptions \ref{assn:new_combined_all} (a) and \ref{assn:new_combined_all} (c),  we have
$$
\left\|E_n[\eta p_{1:k}(X)p_{1:k}(X)^T]\right\|_2=O_{P}\left(\sqrt{\frac{\xi_k^2\log k}{n}}\right).
$$
\end{lemma}

\begin{lemma}\label{lemma:cov_est}
Under Assumptions \ref{assn:new_combined_all} (a) and \ref{assn:new_combined_all} (c), we have
$$
\left\|E_n[\hat{\omega}\hat{\omega}^T]^{-1/2}-{E[\omega\omega^T]}^{-1/2}\right\|_2=O_{P}\left((n^{1/\nu}\vee \ell_k c_k )\sqrt{\frac{\xi_k^2\log k}{n}}\right).
$$
\end{lemma}

\begin{proof}[Proof of Theorem \ref{theorem:main_size}]
First, we are going to show that
$\left\|E_n[\hat{\omega}\hat{\omega}^T]^{-1/2}E_n[\omega]\right\|_{\infty}\leq {cv}$ implies $\theta_0\in {CR}_{\theta}$.
By Assumption \ref{assn:approx_error} for $\mathbf{A}_1$, we have
$$
[\mathbf{A}_1p_{1:k}^T](w_1)\bar{\beta}
\leq 
[\mathbf{A}_1g_0](w_1)+|[\mathbf{A}_1(g_0-p_{1:k}^T\bar{\beta})](w_1)|
\leq
\delta_1(w_1)
$$
for every $w_1\in\mathcal{W}_1$.
By Assumption \ref{assn:approx_error} for $\mathbf{A}_0$, we have
$$
[\mathbf{A}_0p_{1:k}^T](w_0)\beta-\delta_0(w_0)\leq \theta_0(w_0)\leq [\mathbf{A}_0p_{1:k}^T](w_0)\beta+\delta_0(w_0)
$$
for every $w_0\in\mathcal{W}_0$.
Together with $\left\|E_n[\hat{\omega}\hat{\omega}^T]^{-1/2}E_n[\omega]
\right\|_{\infty}\leq {cv}$, we have $\theta_0\in {CR}_{\theta}$. 

The rest of the proof is going to establish
$$
\liminf_{n\rightarrow\infty}\mathbb{P}\left(\left\|E_n[\hat{\omega}\hat{\omega}^T]^{-1/2}E_n[\omega]
\right\|_{\infty}\leq {cv}\right)\geq 1-\alpha.
$$
We now invoke Lemma \ref{lemma:GaussApprox} under Assumptions \ref{assn:new_combined_all} (a) and \ref{assn:new_combined_all} (b). Observe that as the Gaussian random vectors $Z$ and $Z^\ast$ are centered and share a common covariance matrix, we have $\mathbb{P}\left(\left\|Z\right\|_{\infty}\leq t\right)=\mathbb{P}\left(\left\|Z^\ast\right\|_{\infty}\leq t\right)$.
Hence it holds that
\begin{eqnarray*}
&&
\mathbb{P}\left(\left\|E_n[\hat{\omega}\hat{\omega}^T]^{-1/2}E_n[\omega]
\right\|_{\infty}\leq {cv}\right)
\\
&&\geq
\mathbb{P}\left(\left\|E_n[\hat{\omega}\hat{\omega}^T]^{-1/2}E_n[\eta\hat{\omega}]
\right\|_{\infty}\leq {cv}\right)
\\&&\quad
-\sup_{t}\left|\mathbb{P}\left(\left\|Z^\ast\right\|_{\infty}\leq t\right)
-\mathbb{P}\left(\left\|{E[\omega\omega^T]}^{-1/2}E_n[\eta{\omega}]
\right\|_{\infty}\leq t\right)\right|
\\&&\quad
-\sup_{t}\left|\mathbb{P}\left(\left\|{E[\omega\omega^T]}^{-1/2}E_n[\eta{\omega}]
\right\|_{\infty}\leq t\right)
-\mathbb{P}\left(\left\|E_n[\hat{\omega}\hat{\omega}^T]^{-1/2}E_n[\eta\hat{\omega}]
\right\|_{\infty}\leq t\right)\right|
\\&&\quad
-\sup_{t}
\left|\mathbb{P}\left(\left\|Z\right\|_{\infty}\leq t\right)-\mathbb{P}\left(\left\|{E[\omega\omega^T]}^{-1/2}E_n[\omega]
\right\|_{\infty}\leq t\right)\right|
\\&&\quad
-\sup_{t}
\left|\mathbb{P}\left(\left\|{E[\omega\omega^T]}^{-1/2}E_n[\omega]
\right\|_{\infty}\leq t\right)-\mathbb{P}\left(\left\|E_n[\hat{\omega}\hat{\omega}^T]^{-1/2}E_n[\omega]
\right\|_{\infty}\leq t\right)\right|.
\end{eqnarray*}
Following its definition, $\mathbb{P}\left(\left\|E_n[\hat{\omega}\hat{\omega}^T]^{-1/2}E_n[\eta\hat{\omega}]
\right\|_{\infty}\leq {cv}\right)=1-\alpha$.
By Lemma \ref{lemma:GaussApprox}, it suffices to show 
\begin{equation}\label{eq:new_goal1}
\sup_{t}
\left|\mathbb{P}\left(\left\|{E[\omega\omega^T]}^{-1/2}E_n[\omega]
\right\|_{\infty}\leq t\right)-\mathbb{P}\left(\left\|E_n[\hat{\omega}\hat{\omega}^T]^{-1/2}E_n[\omega]
\right\|_{\infty}\leq t\right)\right|=o(1)
\end{equation}
and 
\begin{equation}\label{eq:new_goal2}
\sup_{t}\left|\mathbb{P}\left(\left\|{E[\omega\omega^T]}^{-1/2}E_n[\eta{\omega}]
\right\|_{\infty}\leq t\right)
-\mathbb{P}\left(\left\|E_n[\hat{\omega}\hat{\omega}^T]^{-1/2}E_n[\eta\hat{\omega}]
\right\|_{\infty}\leq t\right)\right|=o(1).
\end{equation}
We can bound the first probability as follows: 
\begingroup
\allowdisplaybreaks
\begin{eqnarray}
&&
\sup_{t}\left|\mathbb{P}\left(\left\|{E[\omega\omega^T]}^{-1/2}E_n[{\omega}]
\right\|_{\infty}\leq t\right)
-\mathbb{P}\left(\left\|E_n[\hat{\omega}\hat{\omega}^T]^{-1/2}E_n[\omega]
\right\|_{\infty}\leq t\right)\right|
\notag\\
&\leq&
\sup_{t}\mathbb{P}\left(\left|\left\|{E[\omega\omega^T]}^{-1/2}E_n[{\omega}]
\right\|_{\infty}-t\right|\leq 1/{(\sqrt{n}\log k)}\right)
\notag\\&&
+\mathbb{P}\left(\left|\left\|{E[\omega\omega^T]}^{-1/2}E_n[{\omega}]
\right\|_{\infty}-\left\|E_n[\hat{\omega}\hat{\omega}^T]^{-1/2}E_n[\omega]
\right\|_{\infty}\right|>1/{(\sqrt{n}\log k)}\right)
\notag\\
&\leq&
\sup_{t}\mathbb{P}\left(\left|\left\|Z\right\|_{\infty}-t\right|\leq 1/{(\sqrt{n}\log k)}\right)
\notag\\&&+
2\sup_{t}\left|\mathbb{P}\left(\left\|Z\right\|_{\infty}\leq t\right)
-\mathbb{P}\left(\left\|E_n[{E[\omega\omega^T]}^{-1/2}{\omega}]
\right\|_{\infty}\leq t\right)\right|
\notag\\&&
+\mathbb{P}\left(\left|\left\|{E[\omega\omega^T]}^{-1/2}E_n[{\omega}]
\right\|_{\infty}-\left\|E_n[\hat{\omega}\hat{\omega}^T]^{-1/2}E_n[\omega]
\right\|_{\infty}\right|>1/{(\sqrt{n}\log k)}\right)
\notag\\
&\leq&
o(1)
+\mathbb{P}\left(\left|\left\|{E[\omega\omega^T]}^{-1/2}E_n[{\omega}]
\right\|_{\infty}-\left\|E_n[\hat{\omega}\hat{\omega}^T]^{-1/2}E_n[\omega]
\right\|_{\infty}\right|>1/{(\sqrt{n}\log k)}\right),\label{eq:anticonent}
\end{eqnarray}
\endgroup
where the last inequality uses Lemma \ref{lemma:GaussApprox} and an anti-concentration argument, which implies that 
$$
\sup_{t}\mathbb{P}\left(\left|\left\|Z\right\|_{\infty}-t\right|\leq 1/{(\sqrt{n}\log k)}\right)=o(1).
$$
To see how the anti-concentration argument works, observe that
\begin{align*}
&\sup_{t}\mathbb{P}\left(\left|\left\|Z\right\|_{\infty}-t\right|\leq 1/{(\sqrt{n}\log k)}\right)\\
\le& 
\sup_{z\in \mathbb R^k}\mathbb{P}\left(z<Z\leq z+1/{(\sqrt{n}\log k)}\right)+ \sup_{z\in \mathbb R^k}\mathbb{P}\left(z-1/{(\sqrt{n}\log k)}\le Z\le  z\right).
\end{align*}
Then the Nazarov's anti-concentration inequality (Lemma A.1 in \cite{CCK2017}) implies that the first term on the right hand side
\begin{align*}
&\sup_{z\in \mathbb R^k}\mathbb{P}\left(z<Z\leq z+1/{(\sqrt{n}\log k)}\right)
\le  C(n\log k)^{-1/2}=o(1),
\end{align*}
where $C$ is a constant that depends only on $b$ from Assumption \ref{assn:new_combined_all} (b). The second term on the right hand side above follows a similar argument.
Now, for the remaining term in Equation (\ref{eq:anticonent}), note that  
\begin{align*}
\left|\left\|{E[\omega\omega^T]}^{-1/2}E_n[{\omega}]
\right\|_{\infty}-\left\|E_n[\hat{\omega}\hat{\omega}^T]^{-1/2}E_n[\omega]
\right\|_{\infty}\right|
\leq&
\left\|(E_n[\hat{\omega}\hat{\omega}^T]^{-1/2}-{E[\omega\omega^T]}^{-1/2})E_n[\omega]
\right\|_{\infty}
\\
\leq&
\left\|E_n[\hat{\omega}\hat{\omega}^T]^{-1/2}-{E[\omega\omega^T]}^{-1/2} \right\|_2\left\|E_n[\omega]
\right\|_2\\
=& O_{P}\left((n^{1/\nu}\vee\ell_kc_k)\sqrt{\frac{\xi_k^4 \log k}{n^2}}\right)=o_{P}(1)
\end{align*}
follows from Lemma \ref{lemma:conv_omega2}, Lemma \ref{lemma:cov_est}, and Assumption \ref{assn:new_combined_all} (c)-(iv).  This verifies Equation \eqref{eq:new_goal1}. 

We next show Equation \eqref{eq:new_goal2}. 
In a similar way to Equation \eqref{eq:anticonent}, we can bound 
\begin{eqnarray*}
&&
\sup_{t}\left|\mathbb{P}\left(\left\|{E[\omega\omega^T]}^{-1/2}E_n[\eta{\omega}]
\right\|_{\infty}\leq t\right)
-\mathbb{P}\left(\left\|E_n[\hat{\omega}\hat{\omega}^T]^{-1/2}E_n[\eta\hat{\omega}]
\right\|_{\infty}\leq t\right)\right|
\\
&\leq&
o(1)
+\mathbb{P}\left(\left|\left\|{E[\omega\omega^T]}^{-1/2}E_n[\eta{\omega}]
\right\|_{\infty}-\left\|E_n[\hat{\omega}\hat{\omega}^T]^{-1/2}E_n[\eta\hat{\omega}]
\right\|_{\infty}\right|>1/{(\sqrt{n}\log k)}\right). 
\end{eqnarray*}
Note that
\begin{eqnarray*}
&&
\left|\left\|{E[\omega\omega^T]}^{-1/2}E_n[\eta{\omega}]
\right\|_{\infty}-\left\|E_n[\hat{\omega}\hat{\omega}^T]^{-1/2}E_n[\eta\hat{\omega}]
\right\|_{\infty}\right|
\\
&\leq&
\left\|(E_n[\hat{\omega}\hat{\omega}^T]^{-1/2}-{E[\omega\omega^T]}^{-1/2})E_n[\eta{\omega}]\right\|_{\infty}
+\left\|(E_n[\hat{\omega}\hat{\omega}^T]^{-1/2}-{E[\omega\omega^T]}^{-1/2})E_n[{\omega}]\right\|_{\infty}
\\&&+\left\|(E_n[\hat{\omega}\hat{\omega}^T]^{-1/2}-{E[\omega\omega^T]}^{-1/2})E_n[\eta(\hat{\omega}-{\omega})]
\right\|_{\infty}+\left\|{E[\omega\omega^T]}^{-1/2}E_n[\eta(\hat{\omega}-{\omega})]
\right\|_{\infty}\\
&\leq&
\left\|E_n[\hat{\omega}\hat{\omega}^T]^{-1/2}-{E[\omega\omega^T]}^{-1/2}\right\|_2\left\|E_n[\eta{\omega}]\right\|_2
+\left\|E_n[\hat{\omega}\hat{\omega}^T]^{-1/2}-{E[\omega\omega^T]}^{-1/2}\right\|_2\left\|E_n[{\omega}]\right\|_2
\\&&+\left(\left\|(E_n[\hat{\omega}\hat{\omega}^T]^{-1/2}-{E[\omega\omega^T]}^{-1/2})\right\|_2+\left\|{E[\omega\omega^T]}^{-1/2}\right\|_2\right)\left\|E_n[\eta(\hat{\omega}-{\omega})]
\right\|_2\\
&\leq& O_{P}\left((n^{1/\nu}\vee\ell_kc_k)\sqrt{\frac{\xi_k^4 \log k}{n^2}}\right)+O_{P}(1)\left\|E_n[\eta(\hat{\omega}-{\omega})]
\right\|_2\\
&=&
o(1)
\end{eqnarray*}
follows from Lemma \ref{lemma:conv_omega2}, Lemma \ref{lemma:conv_Qast}, Lemma \ref{lemma:cov_est}, and the fact that with probability $1-o(1)$,
\begin{eqnarray*}
\left\|E_n[\eta(\hat{\omega}-{\omega})]\right\|_2
&=&
\left\|(E_n[\eta p_{1:k}(X)p_{1:k}(X)^T])E_n\left[p_{1:k}(X)p_{1:k}(X)^T\right]^{-1}E_n[\omega]\right\|_2\\
&=&
\left\|E_n[\eta p_{1:k}(X)p_{1:k}(X)^T]\right\|_2\|E_n\left[p_{1:k}(X)p_{1:k}(X)^T\right]^{-1}\|_2\left\|E_n[\omega]\right\|_2\\
&=&
O\left( \sqrt{\frac{\xi_k^4 \log k}{n^2}}\right)\\
&=&
o(1).
\end{eqnarray*}
Note that we have used $\|E_n\left[p_{1:k}(X)p_{1:k}(X)^T\right]^{-1}\|_2 =O_{P}(1)$. To see this, observe that
$\|E_n\left[p_{1:k}(X)p_{1:k}(X)^T\right]-E[p_{1:k}(X)p_{1:k}(X)^T]\|=o_{P}(1)$ following Lemma 6.2 in \cite{belloni2015some} under 
Assumption \ref{assn:new_combined_all} (c)-(iv) . Therefore, all eigenvalues of $E_n\left[p_{1:k}(X)p_{1:k}(X)^T\right]$ are bounded away from zero with probability approaching one following the same argument in the proof of Lemma \ref{lemma:cov_est}. This verifies Equation (\ref{eq:new_goal2}). 
\end{proof}

\section{Proofs for the Auxiliary Lemmas}\label{sec:proofs_auxiliary_lemmas}
This Section contains the proofs of the lemmas in Appendix \ref{sec:theorem:main_size}.
\subsection{Proof of Lemma \ref{lemma:GaussApprox}}
\begin{proof}
Observe that $E[\omega]=0$. The first uniform convergence in probability follows from Proposition 2.1 in \cite{chernozhukov2017} under their Conditions (M.1), (M.2), and (E.2), that are implied by our Assumption \ref{assn:new_combined_all} (b). 
The second follows from the same proposition in \cite{chernozhukov2017} -- note that Conditions (M.1), (M.2), and (E.2) and the independence between $\eta$ and the data imply
$E[( \eta(E[\omega\omega^T]^{-1/2})_j\omega)^2]\ge b$,
$E[| \eta(E[\omega\omega^T]^{-1/2})_j\omega|^{2+\kappa}]\le B_n^\kappa$,  and 
$E[ \|\eta E[\omega\omega^T]^{-1/2}\omega\|_\infty^q]\le B_n^{q}$.
Finally, the statement on covariance equality is implied by the first two statements,  Proposition 2.1 in \cite{chernozhukov2017} and the equality $E[E[\omega\omega^T]^{-1/2}\omega(E[\omega\omega^T]^{-1/2}\omega)^T]=E[\eta^2E[\omega\omega^T]^{-1/2}\omega(E[\omega\omega^T]^{-1/2}\omega)^T]$.
\end{proof}
\subsection{Proof of Lemma \ref{lemma:conv_omega2}}
\begin{proof}
By Jensen's inequality, we have
\begingroup
\allowdisplaybreaks
\begin{eqnarray*}
E[\left\|E_n[{\omega}]\right\|_2]
&=&
E[(E_n[{\omega}]^TE_n[{\omega}])^{1/2}]\\
&\leq& 
(E[E_n[{\omega}]^TE_n[{\omega}]])^{1/2}\\
&=&
\sqrt{\frac{1}{n}}E[{\omega}^T{\omega}]^{1/2}\\
E[\left\|E_n[(\eta+1){\omega}]\right\|_2]
&=&
E[(E_n[(\eta+1){\omega}]^TE_n[(\eta+1){\omega}])^{1/2}]\\
&\leq& 
(E[E_n[(\eta+1){\omega}]^TE_n[(\eta+1){\omega}]])^{1/2}\\
&=&
\sqrt{\frac{1}{n}}(E[(\eta+1)^2{\omega}^T{\omega}])^{1/2}\\
&=&
\sqrt{\frac{1}{n}}(E[{\omega}^T{\omega}])^{1/2}.
\end{eqnarray*}
\endgroup
Note that we used the independence between $\eta$ and the data.
We can further bound 
\begin{eqnarray*}
E[{\omega}^T{\omega}]^{1/2}
&=&
\left(E[\|p_{1:k}(X)\|_2^2(Y-p_{1:k}(X)^TQ^{-1}E\left[p_{1:k}(X)Y\right])^2]\right)^{1/2}\\
&=&
\left(E[\|p_{1:k}(X)\|_2^2(Y-p_{1:k}(X)^T\bar{\beta})^2]\right)^{1/2}\\
&\leq&
\xi_k\left(E[(Y-p_{1:k}(X)^T\bar{\beta})^2]\right)^{1/2}\\
&\leq&
\xi_k\left(E[Y^2]\right)^{1/2}.
\end{eqnarray*}
Therefore, the statement of the lemma follows.
\end{proof}


\subsection{Proof of Lemma \ref{lemma:conv_Qast}}
\begin{proof}
By the second statement of Lemma 6.1 in \cite{belloni2015some}, we have 
\begin{eqnarray*}
E[\left\|E_n[\eta p_{1:k}(X)p_{1:k}(X)^T]\right\|_2\mid \{Y_i,X_i\}]
&=&O\left(
\sqrt{\frac{\log k}{n}}\left\|\left(E_n[(p_{1:k}(X)p_{1:k}(X)^T)^2]\right)^{1/2}\right\|_2\right).
\end{eqnarray*}
We can further bound the norm part by
\begin{eqnarray*}
\left\|\left(E_n[(p_{1:k}(X)p_{1:k}(X)^T)^2]\right)^{1/2}\right\|_2
&=&
\left\|\left(E_n[(p_{1:k}(X)(p_{1:k}(X)^Tp_{1:k}(X))p_{1:k}(X)^T]\right)^{1/2}\right\|_2\\
&\leq&
\xi_k\|E_n\left[p_{1:k}(X)p_{1:k}(X)^T\right]^{1/2}\|_2.
\end{eqnarray*}
By \citet[Theorem 4.6]{belloni2015some}, we have $\|E_n\left[p_{1:k}(X)p_{1:k}(X)^T\right]^{1/2}\|_2=O_{P}(1)$ under Assumption \ref{assn:new_combined_all} (c).
\end{proof}

\subsection{Proof of Lemma \ref{lemma:cov_est}}
\begin{proof}
By Lemma A.2 of \cite{belloni2015some}, we can bound 
$$
\left\|E_n[\hat{\omega}\hat{\omega}^T]^{-1/2}-{E[\omega\omega^T]}^{-1/2}\right\|_2
\leq
\left\|E_n[\hat{\omega}\hat{\omega}^T]^{-1}-{E[\omega\omega^T]}^{-1}\right\|_2\left\|{E[\omega\omega^T]}\right\|_2^{1/2}.
$$
Observe that by Jensen's inequality, $\{E[\max_{1\le i \le n}|Y_i-g_0(X_i)|^{2}]\}^{1/2}=O(n^{1/\nu})$ under Assumption \ref{assn:new_combined_all} (c)-(i) Applying Theorem 4.6 in \cite{belloni2015some}, we have 
$$
\left\|E_n[\hat{\omega}\hat{\omega}^T]-{E[\omega\omega^T]}\right\|_2=O_{P}\left((n^{1/\nu}\vee \ell_k c_k )\sqrt{\frac{\xi_k^2\log k}{n}}\right)
$$
under Assumptions \ref{assn:new_combined_all} (a) and \ref{assn:new_combined_all} (c).
Notice that $\left\|{E[\omega\omega^T]}^{-1}\right\|_2=O(1)$ and $\left\|{E[\omega\omega^T]}\right\|_2=O(1)$. We now claim that $\|E_n[\hat{\omega}\hat{\omega}^T]^{-1}\|_2=O_{P}(1)$. In fact, all eigenvalues of $E_n[\hat{\omega}\hat{\omega}^T]$ are bounded away from zero. To see this, assume without loss of generality  $E[\omega\omega^T] =I$. Suppose that at least one of eigenvalues of $E_n[\hat{\omega}\hat{\omega}^T]$ is strictly smaller than $1/2$, then there exists a vector $a\in\mathbb R^k$ on the unit sphere such that $a'E_n[\hat{\omega}\hat{\omega}^T] a<1/2$ and thus $\|E_n[\hat{\omega}\hat{\omega}^T]-E[\omega\omega^T]\|_2\ge |a^T(E_n[\hat{\omega}\hat{\omega}^T]-E[\omega\omega^T])a|=|a^TE_n[\hat{\omega}\hat{\omega}^T] a-1|>1/2$, a contradiction. This implies that all eigenvalues of $E_n[\hat{\omega}\hat{\omega}^T]^{-1}$ are bounded from above and thus the claim. Hence, we have
$$
\left\|E_n[\hat{\omega}\hat{\omega}^T]^{-1}-{E[\omega\omega^T]}^{-1}\right\|_2 \le \|E_n[\hat{\omega}\hat{\omega}^T]^{-1}\|_2\left\|E_n[\hat{\omega}\hat{\omega}^T]-E[\omega\omega^T]\right\|_2\|E[\omega\omega^T]^{-1}\|_2,
$$
which, combined with the above bound, yields
$$
\left\|E_n[\hat{\omega}\hat{\omega}^T]^{-1/2}-{E[\omega\omega^T]}^{-1/2}\right\|_2=O_{P}\left((n^{1/\nu}\vee \ell_k c_k )\sqrt{\frac{\xi_k^2\log k}{n}}\right).
$$
Therefore, the statement of the lemma follows.
\end{proof}

\section{Simulation Analysis}\label{sec:simulation}
In this section, we use Monte Carlo simulations to check whether the proposed method works as the theory claims.
Consider the following data generating process.
\begin{align*}
Y(t,x,u) &= 0.5 t - 0.1 x + u
\\
T(x) &= 
\begin{cases}
0.5 x & \text{ if } x < 0
\\
0     & \text{ if } x \geq 0
\end{cases}
\end{align*}
We design this policy schedule $T$ to mimic the actual policy schedule that we use in our empirical analysis in Section \ref{sec:empirical_analysis}.
Allowing for the endogeneity of $X$, we generate $(X,U)$ from the bivariate normal distribution with $E[X]=E[U]=0$, $Var(X)=1.00$, $Cov(X,U)=0.10$ and $Var(U)=0.10$.
In this data generating process, the true partial effect is $h^1(0)=0.5$. 
We experiment with three different sample sizes $n=1000$, $2000$ and $4000$.
We implement the algorithm in Section \ref{sec:empirical_analysis} with the kink location at $0$ and the subsample with $X\in[-1,1]$.
The number of multiplier bootstrap iterations is set to $M=2500$.
We experiment with $k \in \{4,8,12\}$ and set $\delta_0 = \delta_1 = 0.01$ throughout.
Each set of simulations is based on 10,000 Monte Carlo iterations.

Table \ref{tab:lengths} summarizes average lengths and coverage frequencies of the 95\% confidence intervals under alternative shape restrictions across the three different sample sizes, $n=1000$, $2000$ and $4000$.
First, note that the lengths decrease as the sample size $n$ increases for each sieve dimension $k$ and for each set of shape restrictions.
Second, observe that the coverage frequencies are quite close to the nominal probability 95\% for each sieve dimension $k$ and for each set of shape restrictions.
Third, when the sieve dimension takes $k \in \{8,12\}$, the shape restriction \eqref{eq:rkd_shape_restriction_2} contributes to shrinking the average lengths without sacrificing the coverage frequencies.
These results imply that shape restrictions contribute to more informative statistical inference.

\begin{table}
	\centering
		\begin{tabular}{crccccccc}
		\hline\hline
		  Sieve & & \multicolumn{3}{c}{Average Length} && \multicolumn{3}{c}{Coverage}\\
			\cline{3-5}\cline{7-9}
			Dimension & Sample Size $n$ & $1000$ & $2000$ & $4000$ && $1000$ & $2000$ & $4000$\\
			\hline
			 \multirow{2}{*}{$k$=4} & No Shape Restriction &
			0.656 & 0.470 & 0.338 && 0.948 & 0.947 & 0.949\\
			& Shape Restrictions \eqref{eq:rkd_shape_restriction_2} &
			0.647 & 0.470 & 0.338 && 0.948 & 0.947 & 0.949
			\\
			\\
			 \multirow{2}{*}{$k$=8} & No Shape Restriction &
			6.039 & 4.283 & 3.037 && 0.950 & 0.950 & 0.948\\
			& Shape Restrictions \eqref{eq:rkd_shape_restriction_2} &
			3.519 & 2.646 & 2.020 && 0.950 & 0.950 & 0.948\\
			\\
			 \multirow{2}{*}{$k$=12} & No Shape Restriction &
			20.675 & 14.679 & 10.406 && 0.942 & 0.941 & 0.942\\
			& Shape Restrictions \eqref{eq:rkd_shape_restriction_2} &
			10.819 &  7.879 &  5.690 && 0.942 & 0.941 & 0.942\\
		\hline\hline
		\end{tabular}
	\caption{Average lengths and coverage frequencies of the 95\% confidence intervals under alternative shape restrictions. All the results are based on 10,000 Monte Carlo iterations.}
	\label{tab:lengths}
\end{table}

\newpage
\bibliography{mybib}
\end{document}